\documentclass[a4paper, 11pt]{article}
\usepackage{amsthm,amsmath,amssymb}
\usepackage{subfig}
\usepackage[usenames,dvipsnames]{color}
\usepackage[pdftex,breaklinks,colorlinks, citecolor={BlueViolet}, linkcolor={Blue},urlcolor=Maroon]{hyperref}
\usepackage{tikz,pgfkeys}
\usepackage{tkz-graph}
\usetikzlibrary{quotes}
\usepackage{graphicx}
\usepackage{charter,eulervm}%
\usepackage{tabularx,enumerate}
\usetikzlibrary{calc}
\usepackage[final,expansion=alltext,protrusion=true]{microtype}
\usepackage{tikz}

\newenvironment{tightcenter}
 {\parskip=0pt\par\nopagebreak\centering}
 {\par\noindent\ignorespacesafterend}

\usepackage{xargs}
\usepackage{xifthen}
\usepackage{framed}

\usepackage{ctable}

\newlength{\RoundedBoxWidth}
\newsavebox{\GrayRoundedBox}
\newenvironment{GrayBox}[1]%
   {\setlength{\RoundedBoxWidth}{.8\textwidth}
    \def\boxheading{#1}
    \begin{lrbox}{\GrayRoundedBox}
       \begin{minipage}{\RoundedBoxWidth}%
   }{%
       \end{minipage}
    \end{lrbox}%
    \begin{tightcenter}%
    \begin{tikzpicture}%
       \node(Text)[draw=black!20,fill=white,rounded corners,%
             inner sep=2ex,text width=\RoundedBoxWidth]%
             {\usebox{\GrayRoundedBox}};
        \coordinate(x) at (current bounding box.north west);
        \node [draw=white,rectangle,inner sep=3pt,anchor=north west,fill=white] 
        at ($(x)+(6pt,.75em)$) {\boxheading};
    \end{tikzpicture}
    \end{tightcenter}\vspace{0pt}%
    \ignorespacesafterend
}    
\newenvironment{problem}[2][]{\noindent\ignorespaces%
                                \FrameSep=6pt%
                                \parindent=0pt%
                \vspace*{-.5em}
                \ifthenelse{\isempty{#1}}{%
                  \begin{GrayBox}{#2}%
                }{%
                  \begin{GrayBox}{#2 parameterized by~{#1}}%
                }
                \newcommand\Prob{Output:}%
                \newcommand\Input{Input:}%
                \begin{tabular*}{\textwidth}{@{\hspace{.1em}} >{\itshape} p{1.2cm} p{0.85\textwidth} @{}}%
            }{
                \end{tabular*}%
                \end{GrayBox}%
                \vspace*{-.5em}
                \ignorespacesafterend
            }   

\theoremstyle{plain} 
\newtheorem{theorem}{Theorem}[section]
\newtheorem{lemma}[theorem]{Lemma}
\newtheorem{corollary}[theorem]{Corollary}
\newtheorem{proposition}[theorem]{Proposition}

\tikzstyle{filled vertex} = [{circle,blue,draw,fill=black!50,inner sep=1pt}]
\newtheorem{reduction}{Rule}

\newcommand{\tone}{type-\textsc{i}}
\newcommand{\ttwo}{type-\textsc{ii}}

\title{Polynomial Kernels for Paw-free Edge Modification Problems\thanks{To appear in the proceedings of the 16th Annual Conference on Theory and Applications of Models of Computation (TAMC 2020).}}
\author{Yixin Cao\thanks{Department of Computing, Hong Kong Polytechnic University, Hong Kong, China.  {\tt yixin.cao@polyu.edu.hk, yuping.ke@connect.polyu.hk.} }
  \and
  Yuping Ke\footnotemark[1]
  \and
  Hanchun Yuan\thanks{School of Computer Science and Engineering, Central South University, Changsha, China.}
}

\begin{document}
\maketitle

\begin{abstract}
  Let $H$ be a fixed graph.  Given a graph $G$ and an integer $k$, the $H$-free edge modification problem asks whether it is possible to modify at most $k$ edges in $G$ to make it $H$-free.  Sandeep and Sivadasan (IPEC 2015) asks whether the paw-free completion problem and the paw-free edge deletion problem admit polynomial kernels.  We answer both questions affirmatively by presenting, respectively, $O(k)$-vertex and $O(k^4)$-vertex kernels for them.  This is part of an ongoing program that aims at understanding compressibility of $H$-free edge modification problems.
\end{abstract}

\section{Introduction}
A graph modification problem asks whether one can apply at most $k$ modifications to a graph to make it satisfy certain properties. By modifications we usually mean additions and/or deletions, and they can be applied to vertices or edges.
Although other modifications are also considered, most results in literature are on vertex deletion and the following three edge modifications: edge deletion, edge addition, and edge editing (addition/deletion).

Compared to the general dichotomy results on vertex deletion problems \cite{lewis-80-node-deletion-np,cai-96-hereditary-graph-modification}, the picture for edge modification problems is far murkier.  Embarrassingly, this remains true for the simplest case, namely, $H$-free graphs for fixed graphs $H$.  This paper is a sequel to \cite{cao-18-diamond-editing}, and we are aiming at understanding for which $H$, the $H$-free edge modification problems admitting polynomial kernels.  Our current focus is on the four-vertex graphs; see Figure~\ref{fig:4-vertex} (some four-vertex graphs are omitted because they are complement of ones presented here) and Table~\ref{table:1}.\footnote{Disclaimer: Independent of our work, Eiben et al.~\cite{eiben-19-paw-free-editing} obtain similar results for edge modification problems to paw-free graphs.  They are also able to develop a polynomial kernel for the editing problem.}  We refer the reader to \cite{cao-18-diamond-editing} for background of this research and related work.

\begin{figure}[h]
  \centering
\subfloat[$P_4$]{
    \begin{tikzpicture}[every node/.style={filled vertex},scale=.5]
      \node (a) at (-1,0) {};
      \node (c) at (1,0) {};
      \node (b) at (-1,2) {};
      \node (d) at (1,2) {};
      \draw (d) -- (b) -- (a) -- (c);
    \end{tikzpicture}
}
\quad
\subfloat[$C_4$]{
    \begin{tikzpicture}[every node/.style={filled vertex}, scale=.5]
      \node (a) at (-1,0) {};
      \node (c) at (1,0) {};
      \node (b) at (-1,2) {};
      \node (d) at (1,2) {};
      \draw (a) -- (b) -- (d) -- (c) -- (a);
    \end{tikzpicture}
}
\quad
\subfloat[$K_4$]{
    \begin{tikzpicture}[every node/.style={filled vertex}, scale=.5]
      \node (a) at (-1,0) {};
      \node (d) at (1,0) {};
      \node (b) at (-1,2) {};
      \node (c) at (1,2) {};
      \draw (a) -- (b) -- (c) -- (d) -- (a) -- (c) (b) -- (d);
    \end{tikzpicture}
}
\quad
\subfloat[claw]{
    \begin{tikzpicture}[every node/.style={filled vertex},scale=.25]
      \node (a1) at (-3., 0) {};
      \node (v) at (0, 0) {};
      \node (b1) at (3., 0) {};
      \node (c) at (0,3.5) {};
      \draw[] (a1) -- (v) -- (b1);
      \draw[] (v) -- (c);
    \end{tikzpicture}
}
\quad
\subfloat[paw]{
    \begin{tikzpicture}[every node/.style={filled vertex},scale=.25]
      \node (s) at (0,4) {};
      \node (a1) at (-2,0) {};
      \node (b1) at (2,0) {};
      \node (c) at (0,2) {};
      \draw[] (c) -- (s);
      \draw[] (a1) -- (c) -- (b1) -- (a1);
    \end{tikzpicture}
}
\quad
\subfloat[diamond]{
    \begin{tikzpicture}[every node/.style ={filled vertex}, scale = .5]
      \node (x) at (0, 1) {};
      \node (u) at (-1.5, 0) {};
      \node (v) at (1.5, 0) {};
      \node (y) at (0, -1) {};
      \draw (u) -- (x) -- (v) -- (y) -- (u);
      \draw (x) -- (y);
    \end{tikzpicture}
}

  \caption{Graphs on four vertices (their complements are omitted).}
  \label{fig:4-vertex}
\end{figure}
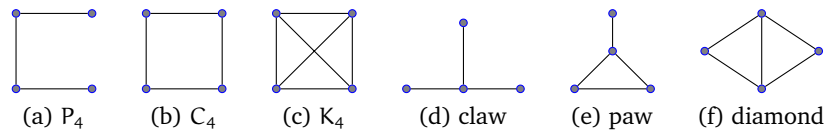

\begin{table}[ht]
  \centering
  \begin{tabular}{c l l l}
    \toprule
    $H$ & completion$\qquad$ & deletion & editing
    \\ \midrule
    $K_4$ & trivial & $O(k^3)$  & $O(k^3)$ \cite{tsur-19-kt-free-kernel}
    \\
    $P_4$ & $O(k^3)$  & $O(k^3)$  & $O(k^3)$ \cite{guillemot-13-kernel-edge-modification}
    \\
    diamond \; & trivial & $O(k^3)$ & $O(k^8)$ \cite{cao-18-diamond-editing}
    \\
    paw & $O(k)$ & $O(k^4)$ [this paper] & unknown
    \\\midrule
    claw & unknown & unknown & unknown
    \\\midrule
    $C_4$ & no  & no  & no \cite{guillemot-13-kernel-edge-modification}
    \\
    \bottomrule
  \end{tabular}
  \caption{The compressibility results of $H$-free edge modification problems for $H$ being four-vertex graphs.  Note that every result holds for the complement of $H$; e.g., the answers are also no when $H$ is $2 K_2$ (the complement of $C_4$).}
  \label{table:1}
\end{table}

In this paper, we show polynomial kernels for both the completion and edge deletion problems when $H$ is the paw (Figure~\ref{fig:4-vertex}(e)).  They answer open problems posed by Sandeep and Sivadasan~\cite{sandeep-15-diamond-free-edge-deletion}.
\begin{theorem}\label{thm:completion}
  The paw-free completion problem has a $38k$-vertex kernel.
\end{theorem}
\begin{theorem}\label{thm:deletion}
 The paw-free edge deletion problem has an $O(k^4)$-vertex kernel.
\end{theorem}

It is easy to see that each component of a paw-free graph is either triangle-free or a complete multipartite graph with at least three parts~\cite{olariu-88-paw-fee}.  This simple observation motivates us to take the modulator approach.  Here by modulator we mean a set of vertices that intersect every paw of the input graph by at least two vertices.  Note that the deletion of all the vertices in the modulator leaves the graph paw-free.  We then study the interaction between the modulator $M$ and the components of $G - M$, which are triangle-free or complete multipartite.
We use slightly different modulators for the two problems under study.

\section{Paw-free graphs}
All graphs discussed in this paper are undirected and simple.  A graph $G$ is given by its vertex set $V(G)$ and edge set $E(G)$.
For a set $U\subseteq V(G)$ of vertices, we denote by $G[U]$ the subgraph induced by $U$, whose vertex set is $U$ and whose edge set comprises all edges of $G$ with both ends in $U$.  We use $G - X$, where $X\subseteq V(G)$, as a shorthand for $G[V(G)\setminus X]$, which is further shortened as $G - v$ when $X = \{v\}$.
For a set $E_+$ of edges, we denote by $G + E_+$ the graph obtained by adding edges in $E_+$ to $G$,---its vertex set is still $V(G)$ and its edge set becomes $E(G)\cup E_+$.  The graph $G - E_-$ is defined analogously.
A paw, as shown in Figure~\ref{fig:4-vertex}(e), is a four-vertex graph with one degree-three vertex and two degree-two vertices.

\begin{problem}{Paw-free Completion}
  \Input & {A graph $G$ and a nonnegative integer $k$.}
  \\
  \Prob  & {Is there a set $E_+$ of at most~$k$ edges such that $G + E_+$  is paw-free?}
\end{problem}

\begin{problem}{Paw-free Edge Deletion}
  \Input & {A graph $G$ and a nonnegative integer $k$.}
  \\
  \Prob  & {Is there a set $E_-$ of at most~$k$ edges such that $G - E_-$  is paw-free?}
\end{problem}

An \emph{independent set} of a graph is a vertex set of pairwise nonadjacent vertices.
For a positive integer $k$, a \textit{$k$-partite} graph is a graph whose vertices  can be partitioned into $k$ different independent sets, called parts, and
a $k$-partite graph is \textit{complete} if all the possible edges are present, i.e.,  there is an edge between every pair of vertices from different parts.
A\textit{ complete multipartite graph} is a graph that is complete $k$-partite for some $k\ge 3$.  Note that here we exclude complete bipartite graphs for convenience.

\begin{proposition}\label{lem:paw-free}
  A graph $G$ is paw-free if and only if every component of $G$ is triangle-free or complete multipartite.
\end{proposition}

In other words, if a connected paw-free graph contains a triangle, then it is necessarily a complete multipartite graph.  Another simple fact is on the adjacency between a vertex and a (maximal) clique in a paw-free graph.

\begin{proposition} \label{pps_v-clique}
Let $K$ be a clique in a paw-free graph.  If a vertex $v$ is adjacent to $K$, then $|K\setminus N[v]| \leq 1$.
\end{proposition}
\begin{proof}
  It is vacuously true when $|K| \le 2$.  For $|K| \ge 3$, any two vertices in $K\setminus N[v]$, a vertex in $K\cap N(v)$, and $v$ itself induce a paw if $|K\setminus N[v]| >1$.
\end{proof}

A set $M \subseteq V(G)$ of vertices is a \emph{modulator} of a graph $G$ if every paw in $G$ intersects $M$ by at least two vertices.  Note that $G- M$ is paw-free.
The following three propositions characterize the interaction between the modulator $M$ and the components of $G - M$.

\begin{proposition}\label{M_partitionbig}
  Let $M$ be a modulator of $G$.  If a vertex $v\in M$ forms a triangle with some component $C$ of $G- M$, then all the neighbors of $v$ are in $M$ and $C$.
\end{proposition}
\begin{proof}
  Suppose that  $u v w$ is a triangle with $u, w\in C$.  If $v$ is also adjacent to vertex $x$ in a different component $C'$ of $G - M$, then $\{u, v, w, x\}$ induces a paw, contradicting the definition of the modulator.
\end{proof}
In other words, if a vertex $v$ in $M$ forms a triangle with a component of $G - M$, then $v$ is a ``private'' neighbor of this component.  As we will see, these components (forming triangles with a single vertex from $M$) are the focus of our algorithms.

\begin{proposition}\label{lem:complete-to-triangle-free}
    Let $G$ be a graph and $M$ a modulator of $G$.  If a vertex $v\in M$ forms a triangle with an edge in a triangle-free component $C$ of $G - M$, then (i) $v$ is adjacent to all the vertices of $C$; and (ii) $C$ is complete bipartite.
\end{proposition}
\begin{proof}
  Let $u v w$ be the triangle with $u, w\in C$.  Suppose that there is a vertex in $C$ nonadjacent to $v$.  We can find a path in $C$ from this vertex to $u$, and let $x$ be the last vertex on this path nonadjacent to $v$.  If $x$ is the neighbor of $u$ on this path, we extend the path by appending $w$ after $u$; otherwise, at least two immediate successors of $x$ on this path are adjacent to $v$.  Therefore, we always have three consecutive vertices $x, u', w'$, of which only $x$ is nonadjacent to $v$.  Since $C$ is triangle-free, $x w'\not\in E(G)$.  But then $\{x, u', v, w'\}$ induces a paw with only one vertex in $M$, a contradiction to the definition of the modulator.  Therefore, every vertex in $C$ is adjacent to $v$.

  For the second assertion, suppose that there are three vertices in $C$ with only one edge among them.  Then they induce a paw with $v$, again contradicting the definition of the modulator.  Thus, $C$ is complete bipartite.
\end{proof}

We say that a triangle-free component of $G - M$ is of type \textsc{i} if it forms a triangle with some vertex in $M$, or type \textsc{ii} otherwise.  By Proposition~\ref{lem:complete-to-triangle-free}, for each \tone{} triangle-free component, all its vertices have a common neighbor in $M$.
A component is trivial if it consists of a single vertex. 
Note that all trivial components of $G- M$ are \ttwo{} trianglle-free components.

\begin{proposition}\label{lem_v-part}
  Let $G$ be a graph and $M$ a modulator of $G$.  For any complete multipartite component $C$ of $G - M$ and vertex $v\in M$ adjacent to $C$, the set of vertices in $C$ that are nonadjacent to $v$ is either empty or precisely one part of $C$.
\end{proposition}
\begin{proof}
  Suppose that the parts of $C$ are $U_1$, $\ldots$, $U_p$.  We have nothing to prove if all the vertices in $C$ are adjacent to $v$.  In the following we assume that, without loss of generality, $v$ is adjacent to $u\in U_1$ and nonadjacent to $w\in U_p$.  We need to argue that $v$ is adjacent to all vertices in the first $p-1$ parts but none in the last part.
  Any vertex $x\in U_i$ with $1 < i < p$ makes a clique with $u$ and $w$.  It is adjacent to $v$ by
  the definition of the modulator ($\{u, v, w, x\}$ cannot induce a paw) and Proposition~\ref{pps_v-clique}.
  Now that $v$ is adjacent to some vertex from another part ($p \ge 3$), the same argument implies $U_1\subseteq N(v)$.  To see $U_p\cap N(v) = \emptyset$, note that a vertex $w'\in U_p\cap N(v)$ would form a paw together with $u, v, w$, contradicting the definition of the modulator.
\end{proof}

A \emph{false twin class} of a graph $G$ is a vertex set in which every vertex has the same open neighborhood.  It is necessarily independent.  The following is immediate from Proposition~\ref{lem_v-part}.
\begin{corollary}\label{cor:false-twins}
  Let $M$ be a modulator of $G$, and $C$ a complete multipartite component of $G - M$.  Each part of $C$ is a false twin class of $G$.
\end{corollary}
The preservation of false twins by all minimum paw-free completions may be of independent interest.
\begin{lemma}\label{lem:false-twins}
  Let $G$ be a graph and let $E_+$ be a minimum set of edges such that $G + E_+$ is paw-free.  Any false twin class of $G$ remains a false twin class of $G + E_+$.
\end{lemma}
\begin{proof}
  Let $I$ be a false twin class of $G$.  If $I$ consists of the isolated vertices, then they remain isolated in $G + E_+$ because $E_+$ is minimum.   Hence it remains false twin class.  In the rest we may assume that $G$ is connected; otherwise we consider the component of $G$ that contains $I$.  Moreover, if $G$ is paw-free, then $E_+$ is empty and the claim holds vacuously.  Now that $G$ is connected and contains a paw, by Proposition~\ref{lem:paw-free}, $G+E_+$ is complete multipartite.  It suffices to show that vertices in $I$ are in the same part of $G+E_+$.  Suppose for contradiction that $u, v\in I$ are in different parts of $G+E_+$; let $P_u$ and $P_v$ be the parts of $G+E_+$ that contain $u$ and $v$ respectively, and assume without loss of generality that $|P_u| \ge |P_v|$.

  Observe that every vertex in $P_u\setminus\{u\}$ is nonadjacent to $u$ in $G+ E_+$, hence nonadjacent to $u$ in $G$, which further implies that it is nonadjacent to $v$ in $G$ because $u$ and $v$ are false twins.  Then $G + (E_+\setminus E_1) \cup E_2$, where $E_1 = \{v x \mid x \in P_u\}$ and $E_2 = \{v x \mid x \in P_v \setminus \{v\}\}$, is also paw-free: Parts $P_u$ and $P_v$ are replaced by $P_u\cup \{v\}$ and $P_v\setminus \{v\}$, while the others remain unchanged.  Since $|E_1| = |P_u| > |P_v| - 1 = |E_2|$, this contradicts that $E_+$ is minimum.  This concludes the proof.
\end{proof}
Since all of our reduction rules are very simple and obvious doable in polynomial time, we omit the analysis of their running time.

\section{Paw-free completion}\label{sec:completion}

The safeness of our first rule is straightforward. 
\begin{reduction}\label{rule:paw-free-components}
If a component of $G$ contains no paw,  delete it.
\end{reduction}

Behind our kernelization algorithm for the paw-free completion problem is the following simple and crucial observation.  After Rule~\ref{rule:paw-free-components} is applied, each remaining component of $G$ contains a paw, hence a triangle, and by Proposition~\ref{lem:paw-free}, we need to make it complete multipartite.
We say that a vertex $v$ and an edge $x y$ \emph{dominate} each other if at least one of $x$ and $y$ is adjacent to $v$.  Note that an edge dominates, and is dominated by, both endpoints of this edge.
Every edge in a complete multipartite graph dominates all its vertices, and hence in a yes-instance, every edge ``almost'' dominates vertices in the component.

\begin{lemma}\label{lem:edge-almost-dominating}
  Let $G$ be a connected graph containing a paw and $uv$ an edge in $G$.  We need to add at least $|V(G) \setminus N[\{u, v\}]| $ edges incident to $u$ or $v$ to $G$ to make it paw-free.
\end{lemma}
\begin{proof}
  Let $U = V(G) \setminus N[\{u, v\}]$, and let $E_+$ be a set of edges such that $G + E_+$ is paw-free.  By Proposition~\ref{lem:paw-free}, $G+E_+$ is a complete multipartite graph. Let $U'$ be the set of vertices in $U$ that are in the same part with $u$ in $G + E_+$.  Then $E_+$ must contain all the edges between $u$ and $U\setminus U'$.  On the other hand, since $u v$ is an edge, $v$ is not in the same part with $u$.  Thus, $E_+$ contains all the edges between $v$ and $U'$, and $|E'_+| \ge |U'| + |U\setminus U'| = |U|$.
 \end{proof}

For the paw-free completion problem, we build the modulator using the procedure in Figure~\ref{alg:modulator}, whose correctness is proved in Lemma~\ref{lem:completion-modulator}.  Starting from an empty set of paws, we greedily add paws: If a paw does not intersect any previously chosen paw with two or more vertices, then add it.  All the vertices of the selected paws already satisfy the definition of the modulator.
After that, we have two more steps, taking all the degree-one vertices of all paws in $G$, and deleting a vertex from $M\cap G'$ for certain component $G'$ of $G$.  Their purposes are to simplify the disposal of triangle-free components of $G - M$: In particular, (iii) and (iv) of Lemma~\ref{lem:completion-modulator} are instrumental for dealing with, respectively, \tone{} and \ttwo{} triangle-free components of $G - M$.

\begin{figure}[h]
  \tikzset{
  algorithm node/.style={draw=gray!50, text width=.8\textwidth, rectangle, rounded corners, inner xsep=20pt, inner ysep=10pt}
}

\centering
\tikz\path
  (0,0) node[algorithm node] {
    \begin{minipage}[t!]{\textwidth}
      \small
      \begin{tabbing}
        Aaa\=AA\=AA\=Aa\=Aa\=MMMMAAAAAAAAAAAAa\=A \kill
        0. \> $\mathcal{F} \leftarrow \emptyset$;  $M\leftarrow \emptyset$;
        \\
        1. \> {\bf for each} paw $F$ of $G$ {\bf do}
        \\
        1.1. \>\> {\bf if} $|F\cap F'|\le 1$ for each paw $F'$ in $\mathcal{F}$ {\bf then}
        \\
        \>\>\> $\mathcal{F} \leftarrow \mathcal{F}\cup \{F\}$;
        \\
        \>\>\> add the vertices of $F$ to $M$;
        \\
        1.2. \>\> {\bf else} add the degree-one vertex of $F$ to $M$;
        \\
        2.\> {\bf for each} component $G'$ of $G$ {\bf do}
        \\
        2.1. \>\> {\bf if} an isolated vertex $v$ of $G' - M$ dominates all the edges in $G'$ {\bf then}
        \\
        \>\>\> find an edge $u w$ in $G[N(v)]$;
        \\
        \>\>\> remove $u$ from $M$;
        \\
        3.\> {\bf return} $M$.
      \end{tabbing}
    \end{minipage}
  };
  \caption{The construction of the modulator for $G$.}
  \label{alg:modulator}
\end{figure}

\begin{lemma}\label{lem:completion-modulator}
  Let $(G,k)$ be an instance of the paw-free completion problem.  The vertex set $M$ constructed in Figure~\ref{alg:modulator} has the following properties.
  \begin{enumerate}[(i)]
  \item The construction is correct and its result is a modulator of $G$.
  \item For each component $G'$ of $G$,  we need to add at least $|M \cap G'| /4$ edges to $G'$ to make it paw-free.
  \item Let  $C$ be a triangle-free component of $G- M$.  If $C$ is nontrivial and any vertex in $C$ is contained in a triangle, then $C$ is of type \textsc{i}.
  \item For each isolated vertex $v$ in $G - M$, there is an edge in $G_v-N[v]$, where $G_v$ is the component of $G$ containing $v$.
  \end{enumerate}
\end{lemma}
\begin{proof}
  We may assume without loss of generality that $G$ is connected and contains a paw; otherwise it suffices to work on its components that contain paws one by one, because both the construction and all the statements are component-wise.

  We denote by $M'$ the set of vertices added to $M$ in step~1.  Note that it is a modulator of $G$ because vertices added in step~1.1 already satisfy the definition.
  Let $X$ be the set of isolated vertices in $G - M'$ each of which dominates all the edges in $G$.  If $X$ is empty, then step~2 is not run, $M = M'$ and we are done.  In the rest, $X\ne \emptyset$.
  We argue first that $X$ is a false twin class.  Vertices in $X$ are pairwise nonadjacent by definition.  Suppose for contradiction that $N(x_1)\ne N(x_2)$ for $x_1, x_2\in X$, then there is a vertex $v$ in $N(x_1)\setminus N(x_2)$ or in $N(x_2)\setminus N(x_1)$.  But then $x_2$ does not dominate edge $v x_1$, or $x_1$ does not dominate edge $v x_2$, contradicting the definition of $X$.  We then argue that any vertex $x\in X$ is in a triangle.   By assumption, $G$ contains a triangle $u v w$.  If $x\in \{u,v,w\}$, then we are done.  Otherwise, $x$ must be adjacent to at least two of $\{u,v,w\}$ to dominate all the three edges in this triangle.  Note that $N(x) \in M'$ because $x$ is isolated in $G - M'$.  This justifies step~2.1 of the construction of $M$.  Note that it removes only one vertex from $M'$.

 Now we prove by contradiction that $M$ is a modulator of $G$.  Suppose that there is a paw $F$ with $|F \cap M| \le 1$.  By construction, $|F \cap M'| \ge 2$, which means $|F \cap M| = 1$ and the only vertex in $M' \setminus M$ is in $F$.   Let $\{v\} = M' \setminus M$ and $\{u\} = F \cap M$; note that the other two vertices of $F$ are in $V(G)\setminus M'$.  Since any vertex in $X$ is isolated in $G - M'$ and dominates all the edges of $G$, every component of $G - M'$ is trivial, which means that the two vertices in $F\setminus\{u, v\}$ are not adjacent.  Therefore, one of $u$ and $v$ must be the degree-three vertex of $F$, and the other is a degree-two vertex of $F$.  But the degree-one vertex of $F$ has been added to $M'$ in step 1.1 or 1.2, a contradiction.  This justifies (i).

 Let $U_1$ and $U_2$ be the sets of vertices added to $M'$ in steps~1.1 and 1.2 respectively; $U_1\cup U_2 = M'$.
 For each paw $F$ added in step~1.1, at least one of its missing edges needs to be added to $G$ to make it paw-free.  This edge is not in any previous selected paw $F'$, because we add $F$ only when $|F\cap F'| \le 1$.  Therefore, we need to add at least $|U_1|/4$ edges to $G[U_1]$ to make it paw-free.  On the other hand, each vertex $v$ in $U_2$ is the degree-one vertex of some paw $F$, (it is possible that all other three vertices of $F$ are in $U_1$,) we need to add at least one edge incident to $v$.  Therefore, we need to add at least $|U_2|/2$ edges incident to vertices in $U_2$ to $G$ to make it paw-free.  Note that these two sets of edges we need to add are disjoint.  The total number of edges we need to add to $G$ to make it paw-free is at least
 \[
   |U_1| /4 + |U_2|/2\ge |U_1\cup U_2| /4= |M'|/4 \ge |M|/4.
 \]
 This concludes assertion (ii).

 Assertion (iii) follows from Proposition~\ref{lem:complete-to-triangle-free} if the triangle has two vertices from $C$: Note that the other vertex must be from $M$ because $C$ itself is triangle-free.  Let the vertices in the triangle be $u, v\in M$ and $w\in C$.
 If $C$ contains the vertex in $M'\setminus M$, then $X\subseteq C$ because it is a false twin class, and there is a vertex in $M$ making a triangle with $C$, and it follows from Proposition~\ref{lem:complete-to-triangle-free}.
 Now that $C$ is a nontrivial component of $G - M'$, we can find a neighbor $x$ of $w$ in $C$.  Note that it is adjacent to at least one of $u$ and $v$; otherwise, $x$ is the degree-one vertex of the paw induced by $\{x, u, v, w\}$ and should be in $M'$.  As a result, $x$ is adjacent to at least one of $u$ and $v$, and then we can use Proposition~\ref{lem:complete-to-triangle-free}.

 Assertion (iv)  follows from the construction of $M$ and the fact that $X$ is a false twin class we proved above.
\end{proof}

The following fact is worth mentioning even though it is not explicitly used in our proofs: If a vertex was removed from $M$ in step 2.1 of Figure~\ref{alg:modulator}, then it forms, with its false twin class and possibly other vertices, a \tone{} triangle-free component of $G - M$.

\begin{corollary}\label{cor:modulator}
  If $(G,k)$ is a yes-instance, then $M$ contains at most $4k$ vertices.
\end{corollary}

We proceed only when $|M| \le 4 k$.
A consequence of this modulator is a simple upper bound on the number of vertices in all the \ttwo{} triangle-free components of $G - M$.  Note that all trivial components of $G - M$ are considered here.

\begin{lemma}\label{lem:completion-type-ii}
Let $(G,k)$ be a yes-instance to the paw-free completion problem on which Rule~\ref{rule:paw-free-components} is not applicable, and let $M$ be the modulator of $G$.
The total number of vertices in all the \ttwo{} triangle-free components of $G - M$ is at most $2k$.
\end{lemma}
\begin{proof}
Let $E_+$ be a solution to  $(G, k)$. We show that each vertex in the \ttwo{} triangle-free components of $G-M$ is incident to at least one edge in $E_+$, from which the lemma follows.
If $C$ is a trivial component, then by Lemma~\ref{lem:completion-modulator}(iv), at least one edge in $E_+$ is incident to $v$.
In the rest $C$ is nontrivial.  Let $x$ be any vertex in $C$, and $G'$ the component of $G$ that contains $x$.  By assumption, $G'$ contains a paw, hence some triangle, say, $uvw$.  By Proposition~\ref{lem:paw-free}, we need to make $G'$ into a complete multipartite graph.
  By Lemma~\ref{lem:completion-modulator}(iii), $x$ cannot be involved in any triangle, hence adjacent to at most one vertex in $\{u, v, w\}$.  Therefore, at least one edge between $x$ and this triangle is in $E_+$, and the proof is now complete.
\end{proof}

Hereafter we consider the components $G'$ of $G$ one by one; let $M' = M\cap V(G')$.  If all components of $G'-M'$ are \ttwo{} triangle-free components, then a bound of the size of $V(G')\setminus M'$ is given in Lemma~\ref{lem:completion-type-ii}.  In the rest, at least one component of $G' - M'$ is a \tone{} triangle-free component or a complete multipartite component.
The way we bound $|V(G')\setminus M'|$ for such a component is to show, after applying some reductions, the minimum number of edges we need to add to $G'$ to make it paw-free is linear on $|V(G')\setminus M'|$.
The first one is very straightforward.

\begin{lemma}\label{lem:double-components}
  If two components in $G' - M'$ are not \ttwo{} triangle-free components, then we need to add at least $|V(G')\setminus M'| / 2$ edges to $G'$ to make it paw-free.
\end{lemma}
\begin{proof}
  Let $X$ and $Y$ be the two components of $G' - M'$ that are not \ttwo{} triangle-free components.
  By definition, they are not trivial, and hence from each of them, we can find an edge, say $x_1 x_2$ in $X$ and $y_1 y_2$ in $Y$.
Since vertices in different components of $G' - M'$ are not adjacent, no vertex in $V(G')\setminus M'$ can dominate both $x_1 x_2$ and $y_1 y_2$.
By Lemma~\ref{lem:edge-almost-dominating}, the number of edges we need to add is at least
\[
  |V(G')\setminus(M'\cup X)| + |V(G')\setminus(M'\cup Y)| - 2\ge |V(G')\setminus M'| - 2.
  \]
Since neither of $X$ and $Y$ is trivial, $|V(G')\setminus M'| \ge 4$, and the lemma follows.
\end{proof}

Henceforth, $G' - M'$ has precisely one \tone{} triangle-free component or one complete multipartite component, but not both.  Each part of such a component is an independent set (recall that a \tone{} triangle-free component is complete bipartite by Proposition~\ref{lem:complete-to-triangle-free}).
The next two propositions are on independent sets $I$ of $G$.
The first is about the cost of separating vertices in $I$ into more than one part; it also means that a sufficiently large independent set cannot be separated.
The second states that if each of the vertices in $I$ is adjacent to all the other vertices, then we can remove all but one vertex in $I$ from the graph.

\begin{proposition}\label{lem:completion-big-independent-set}
  Let $G'$ be a connected graph containing a paw, and let $I$ be an independent set of $G'$.  If we do not add all the missing edges between $I$ and $N(I)$, then we need to add at least $|I| - 1$ edges among $I$ to $G'$ to make it paw-free.
\end{proposition}
\begin{proof}
  Let $E_+$ be a set of edges such that $G' + E_+$ is paw-free.  Since $G'$ contains a paw, $G' + E_+$ is a complete multipartite graph by Proposition~\ref{lem:paw-free}.  If all vertices in $I$ are in the same part of $G' + E_+$, then all the missing edges between $I$ and $N(I)$ are in $E_+$ and we are done.  Otherwise, vertices in $I$ are in at least two parts of $G' + E_+$, and the number of edges among them in $E_+$ is at least $|I| - 1$ (when one part has one vertex and the rest are in another part).
 \end{proof}

\begin{proposition}\label{lem:universal-independent-set}
  Let $I$ be an independent set in a component $G'$ of a graph $G$.  If every vertex in $I$ is adjacent to every vertex in $V(G')\setminus I$, then $(G, k)$ is a yes-instance if and only if $(G - (I\setminus\{v\}), k)$ is a yes-instance for any $v\in I$.
  Moreover, if $G - I$ is connected, then $(G, k)$ is a yes-instance if and only if $(G - I, k)$ is a yes-instance.
\end{proposition}
\begin{proof}
  It suffices to show the if direction.  It is vacuously true when $G$ is paw-free, and henceforth we assume otherwise.
  Suppose that $E_+$ is a solution to $(G - I', k)$, where $I' = I\setminus\{v\}$.  Note that $G' - I'$ is connected because every other vertex is adjacent to $v$.  By Proposition~\ref{lem:paw-free},  $(G' - I') + E_+$ is a complete multipartite graph; let $U_1$, $\ldots$, $U_\ell$ be its parts.  Since $v$ is adjacent to every other vertex, it is in a size-one part; without loss of generality, let it be $U_\ell$.  On the other hand, $I$ remains an independent set and adjacent to every vertex in $V(G')\setminus I$ in  $G + E_+$.  Therefore, $G' + E_+$ is a complete multipartite graph with parts $U_1$, $\ldots$, $U_{\ell-1}$, $I$, hence paw-free.  The proof of the second assertion is very similar and omitted.
 \end{proof}

We are now ready to consider \tone{} triangle-free components.

\begin{lemma}\label{lem:reduced-tone}
  Let $C$ be a \tone{} triangle-free component of $G' - M'$ and let $L\uplus R$ be the bipartition of $C$ with $|L|\geq |R|$.  If any  of the following conditions is satisfied, then we need to add at least $|C|/32$ edges to $G'$ to make it paw-free.
\begin{enumerate}[(i)]
  \item $|L|\leq 4|M'|$;
  \item there is an edge in $G'-N[L]$;
  \item $V(G')\neq N[C]$ and $|L|\leq 2|R|$;
  \item there are $|L|/2$ or more missing edges between $L$ and $N(L)$;
  \item $|L| \le |R|+|M'|$ and $G-N[R]$ has an edge; or
  \item $|L| \le |R|+|M'|$ and there are $|R|/2$ or more missing edges between $R$ and $N(R)$.
\end{enumerate}
\end{lemma}
\begin{proof}
  (i) If $|L|\leq 4 |M'|$, then $|C| = |L| + |R| \le 2 |L| \le 8 |M'|$, and it follows from Lemma~\ref{lem:completion-modulator}(ii).
  (ii)  By Lemma~\ref{lem:edge-almost-dominating}, we need to add at least $|L| \ge |C|/2$ edges.
  (iii) Since $C$ is complete bipartite and $|L| \ge |R|$, we can find a matching of size $|R|$ between $L$ and $R$.  By Lemma~\ref{lem:edge-almost-dominating}, for each vertex $v\in V(G')\setminus N[C]$, the number of edges between $v$ and $C$ we need to add is at least $|R| = (2 |R| + |R|)/3 \ge (|L| + |R|)/3 =|C|/3$.
  (iv) By Proposition~\ref{lem:completion-big-independent-set}, we need to add at least $|L|/2 \ge |C|/4$ edges.

  In the rest, (v) and (vi), $|L|\leq |R|+|M'|$.  We may assume none of the previous conditions is satisfied.  Therefore, $|L| >4|M'|$, which means $|L| \le 2|R|$.  Also note that the proofs for these two conditions are almost the same as conditions (ii) and (iv) respectively.
  (v) By Lemma~\ref{lem:edge-almost-dominating}, we need to add at least $|R| \ge |C|/3$ edges.  (vi) By Proposition~\ref{lem:completion-big-independent-set}, we need to add at least $|R|/2 \ge |C|/6$ edges.
\end{proof}

We say that a \tone{} triangle-free component $C$ of $G' - M'$ is \emph{reducible} if none of the conditions in Lemma~\ref{lem:reduced-tone} holds true. 

\begin{reduction}\label{rule:tone}
  Let $C$ be a \tone{} triangle-free component of $G' - M'$ and let $L\uplus R$ be the bipartition of $C$ with $|L|\geq |R|$.  If $C$ is reducible, then add all the missing edges between $L$ and $N(L)$ and all the missing edges between $V(G')\setminus N[L]$ and $N(L)$; decrease $k$ accordingly; and remove all but one vertex from $(V(G')\setminus N[L]) \cup L$.
\end{reduction}
\begin{lemma}\label{lem:rule:tone}
Rule~\ref{rule:tone} is safe.
\end{lemma}
\begin{proof}
  Let $X = V(G')\setminus N[L]$ and $L' = L\cup X$; note that $L'$ is an independent set because condition (ii) of Lemma~\ref{lem:reduced-tone} is not true.  We show the existence of a solution that contains all the missing edges between $L'$ and $N(L)$.  Let $E_+$ be any solution to $G$, we may assume that it does not have edges between $G'$ and other components of $G$. Let $E'=E_+\cap {V(G')\setminus L' \choose~2}$.
  We show that the set $E'_+$, consisting of  $E'$ together with all the missing edges between $L'$ and $N(L)$, is also a solution of $(G, k)$.  By Proposition~\ref{lem:paw-free}, $G'+E_+$ is a complete multipartite graph, then its subgraph $(G'-L')+E'$ is complete bipartite or complete multipartite.  Therefore, $G'+E_+'$ is a complete multipartite graph, where $L'$ forms a single part.
  Since $E_+$ and $E'_+$ have the same set of edges incident to $V(G')\setminus L'$, to show $|E'_+| \le |E_+|$ we may focus on their edges incident to $L'$.

  Case 1, $|L|>|R|+|M'|$.
  Note that $N(L)\subseteq M'\cup R$, and hence $|N(L)| \le |R| + |M'| < |L|$.
We claim that there is a matching from $N(L)$ to $L$ of size $|N(L)|$; otherwise there exists a vertex set $Y\subseteq N(L)$ with $|N(Y)\cap L| < |Y|$, but then the missing edges between $L$ and $N(L)$ is at least $|Y| \cdot |L\setminus N(Y)|> |Y| \cdot (|L| - |Y|) \ge |L| - 1$,  and condition (iv) of Lemma~\ref{lem:reduced-tone} would be true.
As a result, for each vertex $x\in X$, by Lemma~\ref{lem:edge-almost-dominating}, the number of edges in $E_+$ incident to $x$ is at least $|N(L)\setminus N(x)|$.
In summary, the number of edges in $E_+$ incident to $L'$ is at least $\delta + \sum_{x\in X} |N(L)\setminus N(x)|$, while the number of edges in $E'_+$ incident to $L'$ is $\delta + \sum_{x\in X} |N(L)\setminus N(x)|$, where $\delta$ denotes the missing edges between $L$ and $N(L)$.

Case 2, $|L|\leq |R|+|M'|$.  Since $C$ is reducible, none of the conditions
 in Lemma~\ref{lem:reduced-tone} is satisfied.  In particular, $|L| > 4 |M'|$ (condition i), which means $|R|>3|M'|$ and $|L| \le 2|R|$.  We must have $V(G')=N[C]$ (condition iii).  Moreover, by conditions (ii) and (v), both $V(G') \setminus N[L]$ and $V(G') \setminus N[R]$ are independent sets, and by conditions (iv) and (vi), for $Z\in \{L, R\}$, the number of missing edges between $Z$ and $N(Z)$ is smaller than $|Z|/2$.

 If $|P\cap L|<|L|/2$ for every part $P$ in $G'+E_+$, 
 then every vertex $a$ in $L'$ is incident to more than $|L|/2$ edges in $E_+$.  The total number of edges in $E_+$ that are incident to $L'$ is more than
$
   \frac{1}{2}(|L|\cdot |L|/2) +|X|\cdot |L|/2,
$
 while the total number of edges in $E_+'$ that are incident to $L'$ is at most $|L|/2 + |R|/2 + |X| \cdot |M'\setminus X|$.  Thus, $|E_+'|\leq |E_+|$.

Now suppose that there is a part $P$ in $G'+E_+$  with $|P\cap L|\ge |L|/2$.  If $L'=P$, then $|E_+'|=|E_+|$ and we are done.  In the rest $L'\neq P$. 
Since $|P\cap L|\geq |L|/2$ and there are less than $|L|/2$ missing edges between $N(L)$ and $L$, no vertex in $N(L)$ can be in $P$.  Therefore $P\subseteq L'$, and $L'\not \subseteq P$.  If there is precisely one vertex $v$ in $L'\setminus P$, then all the $|L'| - 1$ edges between $v$ and $P$ are in $E_+$.  But the number of missing edges between $v$ and $M'\cup R$ (in $E_+'$) is less than
\[
  |M'| + |R|/2 \le |L|/4 + |L|/2 < |L| - 1\le |L'| - 1.
\]
  Otherwise, suppose that $|L'\setminus P| = t > 1$, then at least $t(|L'| - t) \ge t |L| / 2$ edges among $L'$ are in $E_+$.  But the number of missing edges between $L'\setminus P$ and $M'\cup R$ (in $E_+'$) is less than $t |M'| + |R|/2 \le t |L|/4 + |L|/2 \le t |L| / 2$.

 We have thus proved that it is safe to add all the missing edges between $L'$ and $N(L)$.  After that, every vertex in $L'$ is adjacent to every vertex in $G' - L'$, and the correctness of removing all but one vertex in $L'$ from $G$ follows from Proposition~\ref{lem:universal-independent-set}.
\end{proof}

In the last we consider the complete multipartite components of $G' - M'$.

\begin{lemma}\label{lem:reduced-cm}
    Let $C$ be a complete multipartite component of $G' - M'$, and let $P^*$ be a largest part of $C$.  If any of the following conditions is satisfied, then we need to add at least $|C|/12$ edges to $G'$ to make it paw-free.
\begin{enumerate}[(i)]
\item $|C| \le 3 |M'|$;
\item there is an edge in $G'- N[C]$;
\item $|P^*| >  2|C|/3$ and $G'- N[P^*]$ has an edge;
\item $|P^*| \le 2|C|/3$ and $V(G')\ne N[C]$; or
\item $|P^*| \le 2|C|/3$ and $V(G')= N[C]$, and for every part $P$ of $C$,
  \begin{itemize}
  \item $G' - N[P]$ contains an edge, or
  \item  there are at least $|P|$ missing edges between $V(G')\setminus N[P]$ and $N(P)$.
  \end{itemize}
\end{enumerate}
\end{lemma}
\begin{proof}
  (i) If $|C|\leq 3 |M'|$, then it follows from Lemma~\ref{lem:completion-modulator}(ii).
  (ii) By Lemma~\ref{lem:edge-almost-dominating}, we need to add at least $|C|$ edges.
   (iii) By Lemma~\ref{lem:edge-almost-dominating}, we need to add at least $|P^*| > 2|C|/3$ edges.
  (iv) Let $E_+$ be a minimum set of edges such that $G + E_+$ is paw-free.
  By Corollary~\ref{cor:false-twins} and Lemma~\ref{lem:false-twins}, the vertices in any part of $C$ remain in the same part in $G'+E_+$.
  Since $V(G') \ne N[C]$, there is a vertex $v\in V(G')\setminus N[C]$, and the missing edges between $v$ and all but one part of $C$ are in $E_+$.  Since
$P^*$ is a largest part and $|P^*| \le 2|C|/3$, we need to add at least $|C|/3$ edges.

For (v), we show that we need to add at least $|P|$ edges incident to $P\cup (V(G')\setminus N[P])$.
By Proposition~\ref{lem_v-part} and the fact $V(G')=N[C]$, the sets $V(G')\setminus N[P]$ for different parts are disjoint.  As a result, each edge is counted at most twice, and the total number of edges we need to add is $|C|/2$.
If there is an edge in $G' - N[P]$, then it follows from Lemma~\ref{lem:edge-almost-dominating}.  Now there is no edge in $G' - N[P]$ and there are more than $|P|$ missing edges between $V(G')\setminus N[P]$ and $N(P)$.  Let $X=V(G')\setminus N[P]$, and let $E_+$ be a minimum set of edges such that $G + E_+$ is paw-free.  By Corollary~\ref{cor:false-twins} and Lemma~\ref{lem:false-twins}, the vertices in any part of $C$ remain in the same part in $G'+E_+$.  If all vertices in $X$ are in the same part as $P$ in $G'+E_+$, then all missing edges between $X$ and $N(P)$ are in $E_+$.  Otherwise, there is a vertex $x$ in $X$ that is not in the same part as $P$ in $G'+E_+$, and all edges between $x$ and $P$ are in $E_+$.  In either case, we need to add at least $|P|$ edges incident to $P\cup X$.  This concludes the proof.
\end{proof}

We say that a complete multipartite component $C$ of $G' - M'$ is \emph{reducible} if none of the conditions in Lemma~\ref{lem:reduced-cm} holds true.

\begin{reduction}\label{rule:ttwo}
  Let $C$ be a reducible complete multipartite component of $G' - M'$ and $P^*$ a largest part of $C$.
  \begin{enumerate}[(1)]
  \item If  $|P^*| > 2|C|/3$, then add all the missing edges between $V(G')\setminus N[P^*]$ and $N(P^*)$; decrease $k$ accordingly;
    and remove $(V(G')\setminus N[P^*])\cup P^*$ from $G$.
  \item Otherwise, find
    a part $P$ such that $V(G')\setminus N[P]$ is an independent set and there are less than $|P|$ missing edges between $V(G')\setminus N[P]$ and $N(P)$.  Add all the missing edges between $V(G')\setminus N[P]$ and $N(P)$; decrease $k$ accordingly;
   and remove $P\cup (V(G')\setminus N[P])$ from $G$.
  \end{enumerate}
\end{reduction}

\begin{lemma}\label{lem: rule: onlyCM}
Rule~\ref{rule:ttwo} is safe.
\end{lemma}
\begin{proof}
  Let $E_+$ be a minimum set of edges such that $G + E_+$ is paw-free;
  note that it does not have edges between $G'$ and other components of $G$.  By Proposition~\ref{lem:paw-free}, $G'+E_+$ is a complete multipartite graph.
By Corollary~\ref{cor:false-twins} and Lemma~\ref{lem:false-twins}, vertices in the same part of $C$ remain in the same part in $G'+E_+$.

Case 1: $|P^*| > 2|C|/3$. Let $X= V(G') \setminus N[P^*]$. Since $C$ is reducible, none of conditions in Lemma~\ref{lem:reduced-cm} is satisfied. In particular, $|C|>3|M'|$ (condition i) and $X$ is an independent set (condition iii).  For each vertex $x\in X$, if it is not in the same part of $G'+E_+$ as $P^*$, then all the $|P^*|$ missing edges between $x$ and $P^*$ are in $E_+$.  
   Since
   \[
     |N(P^*)| < |(C\setminus P^*) \cup M'| = |C\setminus P^*|+ |M'|<|C|/3+|C|/3< |P^*|,
   \]
   putting $x$ to the same part as $P^*$ minimizes the number of edges incident to it.   After adding all the missing edges between $X$ and $N(P^*)$, every vertex in the independent set $P^*\cup X$ is adjacent to every vertex in $V(G') \setminus (P^*\cup X) = N(P^*)$.  By Proposition~\ref{lem_v-part}, $G' - P^*\cup X$ is connected (note that $C$ had at least three parts).  It is thus safe to remove $P^*\cup X$ from $G$ by Proposition~\ref{lem:universal-independent-set}.

Case 2: $|P^*|\le 2|C|/3$. Since $C$ is reducible, none of conditions in Lemma~\ref{lem:reduced-cm} is satisfied. In particular, we have $V(G')=N[C]$ (condition iv) and we can find a part $P$ of $C$ such that $V(G')\setminus N[P]$ is an independent set and the missing edges between $V(G')\setminus N[P]$ and $N(P)$ is less than $|P|$ (condition v).
Let $X=V(G')\setminus N[P]$.
For each vertex $x\in X$, if it is not in the same part of $G'+E_+$ as $P$, then all the $|P|$ missing edges between $x$ and $P$ are in $E_+$.  This is more than the total number of missing edges between $X$ and $N(P)$.  Hence, putting $X\cup P$ in the same part minimizes the number of edges incident to it.
After adding all the missing edges between $X$ and $N(P)$, all vertices in the independent set $P\cup X$ are adjacent to $V(G') \setminus (P\cup X) = N(P)$.  By Proposition~\ref{lem_v-part}, $G' - P\cup X$ is connected (note that $C$ had at least three parts).  It is thus safe to remove $P\cup X$ from $G$ by Proposition~\ref{lem:universal-independent-set}.
\end{proof}

We summarize our kernelization algorithm for the paw-free completion problem in Figure~\ref{alg:completion} and use it to prove our main result of this section.

\begin{figure}[h]
  \tikzset{
  algorithm node/.style={draw=gray!50, text width=.8\textwidth, rectangle, rounded corners, inner xsep=20pt, inner ysep=10pt}
}
\centering
  \tikz\path
  (0,0) node[algorithm node] {
    \begin{minipage}[t!]{\textwidth}
      procedure \texttt{reduce$(G, k)$}

      \small
      \begin{tabbing}
        Aaa\=AA\=AA\=Aa\=Aa\=MMMMAAAAAAAAAAAAa\=A \kill
        0. \> {\bf if} $k<0$ {\bf then return} a trivial no-instance;
        \\
        1.\> remove all paw-free components from $G$;
        \\
        2. \> construct modulator $M$;
        \\
        3. \> {\bf if} $|M| > 4k$ {\bf then return} a trivial no-instance;
        \\
        4.\> {\bf if} $> 2k$ vertices in \ttwo{} triangle-free components of $G - M$  {\bf then}
        \\
        4.1.\>\> {\bf return} a trivial no-instance;
        \\
        5.\> {\bf for each} component $G'$ of $G$ {\bf do}
        \\
        5.1.\>\> $M' \leftarrow V(G')\cap M$;
        \\
        5.2. \>\> \textbf{if} 2 components in $G'-M'$ are not \ttwo{} triangle-free components \textbf{then}
        \\
        \>\>\> \textbf{goto} 5;
        \\
        5.3.\>\> {\bf if} $G' - M'$ has a \tone{} triangle-free component $C$ {\bf then}
        \\
        \>\>\> {\bf if} $C$ is reducible {\bf then} apply Rule~\ref{rule:tone} and \textbf{return} \texttt{reduce}$(G, k)$;
        \\
        5.4.\>\> {\bf if} $G' - M'$ has a complete multipartite component $C$ {\bf then}
        \\
        \>\>\> {\bf if} $C$ is reducible {\bf then} apply Rule~\ref{rule:ttwo} and \textbf{return} \texttt{reduce}$(G, k)$;
        \\
        6.\> {\bf if} $|V(G)| \leq 38k$ {\bf then return} $(G, k)$;
        \\
        7.\> {\bf else return} a trivial no-instance.
      \end{tabbing}

    \end{minipage}
  };
  \caption{The kernelization algorithm for the paw-free completion problem.}
  \label{alg:completion}
\end{figure}

\begin{proof}[Proof of Theorem~\ref{thm:completion}]
  We use the algorithm described in Figure~\ref{alg:completion}.
The correctness of steps~0 and 1 follows from the definition of the problem and  Rule~\ref{rule:paw-free-components} respectively.  Steps 2 and 3 are justified by Lemma~\ref{lem:completion-modulator} and Corollary~\ref{cor:modulator}.  Step~4 is correct because of Lemma~\ref{lem:completion-type-ii}, and after that we only need to consider the components of $G-M$ that are not \ttwo{} triangle-free components, which are dealt with in step~5.  The cost of a component of $G$ is the minimum number of edges we need to add to it to make it paw-free.

If two components of $G'-M'$ are not \ttwo{} triangle-free components, then by Lemma~\ref{lem:double-components}, the cost of $G'$ is at least $|V(G')\setminus M| / 2$.  Therefore, there is nothing to do for step 5.2.
Henceforth, $G' - M'$ has precisely one \tone{} triangle-free component or one complete multipartite component, but not both.
The algorithm enters step~5.3 if there is a \tone{} triangle-free component $C$ in $G'-M'$.  If $C$ is reducible, the correctness of Rule~\ref{rule:tone} is given in Lemma~\ref{lem:rule:tone}; otherwise, the cost of $G'$ is at least $|C|/32$ by Lemma~\ref{lem:reduced-tone}.
The algorithm enters step~5.4 if there is a complete multipartite component $C$ in $G'-M'$.  If $C$ is reducible, the correctness of Rule~\ref{rule:ttwo} is given in Lemma~\ref{lem: rule: onlyCM}; otherwise, the cost of $G'$ is at least $|C|/12$ by Lemma~\ref{lem:reduced-cm}.

When the algorithm reaches step~6, neither of Rules~\ref{rule:tone} and~\ref{rule:ttwo} is applicable.  There are at most $4 k$ vertices in $M$, at most $2 k$ vertices in all the \ttwo{} triangle-free components of $G - M$.  On the other hand, for each other vertex, there is an amortized cost of at least $1/32$.  Therefore, if $(G,k)$ is a yes-instance, then the number of vertices is at most $38k$, and this justifies steps 6 and 7.

We now analyze the running time of this algorithm.  When each time the algorithm calls itself in step~5.3 or 5.4, it removes at least one vertex from the graph.  Therefore, the recursive calls can be made at most $n$ times.  On the other hand, each step clearly takes polynomial time.  Therefore, the algorithm returns in polynomial time.
\end{proof}

We spare little effort in getting a better kernel size and more efficient implementation.  For the running time, e.g., we do not need to call the algorithm after each reduction (application of Rules~\ref{rule:tone} and ~\ref{rule:ttwo}).  We believe that, with more careful analysis, one may show that our kernel is between $12 k$ and $16 k$ vertices.  However, to find a kernel of less than $10 k$ vertices for the problem or to implement it in linear time is quit a challenge.

\section{Paw-free edge deletion} \label{sec_2}

For this problem, we construct the modulator in the standard way.  We greedily find a maximal packing of edge-disjoint paws.  
We can terminate by returning ``no-instance'' if there are more than $k$ of them.  Let $M$ denote the set of vertices in all the paws found; we have $|M|\leq 4k$.   It is a modulator because every paw not included shares at least an edge with some chosen one, hence at least two vertices.

The safeness of the following rule is straightforward: If we do not delete this edge, we have to delete a distinct one from each of the paws, hence $k + 1$.
\begin{reduction}\label{rdt_2}
  Let $u v$ be an edge of $G$.  If there exist $k + 1$ paws such that for any pair of them, the only common edge is $u v$, then delete $uv$ from $G$ and decrease $k$ by 1.
\end{reduction}

\subsection{Complete multipartite components}

We first deal with complete multipartite components of $G- M$.
\begin{reduction}\label{rdt_3}
  Let $C$ be a complete multipartite component of $G- M$.  From each part of $C$, delete all but $k+1$ vertices.
\end{reduction}
\begin{lemma}\label{lem_rdt3}
Rule~\ref{rdt_3} is safe.
\end{lemma}
\begin{proof}
  It suffices to show that for each vertex $v$ deleted from a part $P$ of a complete multipartite component of $G- M$, $(G, k)$ is a yes-instance if and only if $(G - v, k)$ is a  yes-instance.  The only if direction is trivial.  For the other direction, let that $E_-$ be a solution to $(G - v, k)$.  We prove by contradiction that $G - E_-$ is paw-free as well.  A paw of $G - E_-$ must contain $v$, and since vertices in $P$ are nonadjacent (in both $G$ and $G - E_-$), there can be at most two vertices in $P\cap F$.

  Case 1, $v$ is the only vertex in $P\cap F$. By Corollary~\ref{cor:false-twins}, each vertex in $P$ has the same neighbors as $v$ in $G$.  Since $|E_-| \leq k$ and $|P\setminus \{v\}| \geq k+1$, there exists at least one vertex $v' \in P\setminus \{v\}$ that has the same adjacency to $F\setminus\{v\}$ as $v$ in $(G - v) - E_-$, a contradiction.

  Case 2, there exists another vertex $u\in P\cap F$. Since $v$ and $u$ are nonadjacent, they must be the degree-one and a degree-two vertices of $F$.  They are false twins in $G$, and no edge incident to $v$ is deleted.  Thus, $u$ is the degree-one vertex of $F$.  Since $|P\setminus F| \geq k$ and $|E_-| \leq k$, of which one is incident to $u$, there must be at least one vertex $v'\in P\setminus F$ whose neighbors in $F$ are same as $v$ in $(G - v) - E_-$, a contradiction.  This concludes the proof.
\end{proof}

\begin{reduction}\label{rdt_4}
  Let $C$ be a complete multipartite component of $G- M$.  Delete all but $k+4$ parts of $C$ that are adjacent to all vertices in $N(C)$.
\end{reduction}
\begin{lemma}\label{lem_rdt4}
Rule~\ref{rdt_4} is safe.
\end{lemma}
\begin{proof}
  It suffices to show that for each part $P$ deleted from a complete multipartite component of $G- M$, $(G, k)$ is a yes-instance if and only if $(G - P, k)$ is a  yes-instance.  The only if direction is trivial.  For the other direction, let that $E_-$ be a solution to $(G - P, k)$.  We prove by contradiction that $G - E_-$ is paw-free as well.  A paw $F$ of $G - E_-$ must intersect $P$.  By Corollary~\ref{cor:false-twins}, vertices in $P$ are false twins in $G$, hence also in $G - E_-$, and thus there cannot be two of them in $F$.
  Of the at least $k+4$ parts of $C$ satisfying the condition of Rule~\ref{rdt_4},  at least $k + 1$ are disjoint from $F$.  At least one of them have the same adjacency to $F$ as $P$.   Thus, there is a paw in $(G - P) - E_-$, a contradiction.  This concludes the proof.
\end{proof}

In passing we point out that all the parts satisfying the condition of Rule~\ref{rdt_4} form a module of $G$.  Moreover, it suffices to keep $k + 2$ parts, but then the argument would be more involved.

\begin{lemma}\label{lem_numofc.m.}
  After Rules~\ref{rdt_3} and \ref{rdt_4} are applied, there are at most $O(k^3)$ vertices in the complete multipartite components of $G - M$.
\end{lemma}
\begin{proof}
  Let $C$ be a complete multipartite component of $G - M$.
  By Proposition~\ref{lem_v-part}, for each vertex $v\in M$, at most one part of $C$ is nonadjacent to $v$.  Therefore, if $C$ has more than $|M| + k + 4$ parts, then Rule~\ref{rdt_4} is applicable.  Therefore, it has at most $5k + 4$ parts.  On the other hand, since Rule~\ref{rdt_3} is not applicable, the number of vertices in each part is at most $k+1$.  Therefore, $|V(C)| \le (k+1)\cdot (5k+4)$.
By Proposition~\ref{pps_v-clique} and Proposition~\ref{M_partitionbig}, a vertex in $M$ can be adjacent to at most one  complete multipartite component of $G - M$.  Therefore, the total number of complete multipartite components of $G - M$ is $4 k$.  The lemma follows.
\end{proof}

\subsection{Triangle-free components}
In the following, we assume that Rule~\ref{rdt_2} is not applicable.
We mark some vertices from each of the triangle-free components that should be preserved, and then remove all the unmarked vertices.
Recall that a triangle-free component of $G - M$ is of type \textsc{i} or type \textsc{ii} depending on whether it forms a triangle with some vertex in $M$.

The following simple observation is a consequence of Proposition~\ref{M_partitionbig} and the definition of \tone{} triangle-free components.
\begin{corollary}\label{lem_M-partition}
If a vertex in $M$ is adjacent to the triangle-free components of $G - M$, then either it is adjacent to precisely one \tone{} triangle-free component, or it is adjacent to only \ttwo{} triangle-free components.
\end{corollary}

By Proposition~\ref{lem:complete-to-triangle-free}, a \tone{} triangle-free component $C$ of $G- M$ is complete bipartite.

\begin{reduction}\label{rule:type-1}
  Let $\cal C$ be all the \tone{} triangle-free components of $G - M$, and let $U = \bigcup_{C\in \cal C} V(C)$.

  \begin{enumerate}[(i)]
  \item For each $S\subseteq M$ with $|S| = 3$ and each $S'\subseteq S$, mark $k+1$ vertices from $\{x\in U\mid N(x)\cap S = S'\}$.
  \item For each $C\in \cal C$ with bipartition $L\uplus R$ do the following.  For each $S\subseteq M$ with $|S| = 2$ and each $S'\subseteq S$, mark $k+3$ vertices from $\{x\in L\mid N(x)\cap S = S'\}$ and $k+3$ vertices from $\{x\in R\mid N(x)\cap S = S'\}$.
  \end{enumerate}
  Delete all the unmarked vertices from $U$.
\end{reduction}

\begin{lemma}\label{lem_mrk_2}
  Rule~\ref{rule:type-1} is safe.
\end{lemma}
\begin{proof}
Let $G'$ be the graph obtained after applying Rule~\ref{rule:type-1}.
If $(G,k)$ is a yes-instance, then $(G',k)$ is a yes-instance.  For the other direction, suppose that $(G',k)$ is a yes-instance, with a solution $E_-$.  We prove by contradiction that $G - E_-$ is paw-free as well.  A paw $F$ in $G - E_-$ contains at least one deleted vertex, because $G' - E_-$ is paw-free, and at most two deleted vertices, because otherwise $F$ is a paw of $G$ and should be in the modulator.

Consider first that $F$ contains only one deleted vertex $x$.  Let $C$ be the triangle-free component of $G - M$ containing it.
If all the other three vertices in $F$ are from $M$, then in step (i) we have marked $k+1$ vertices in $C$ that have the same adjacency to $F\setminus \{x\}$ as $x$ in $G$.  Since $|E_-| \le k$, the adjacency between $F\setminus \{x\}$ and at least one of these marked vertex is unchanged.  This vertex forms a paw with $F\backslash \{x\}$ in $G'-E_-$, a contradiction.
Now at most two vertices of $F$ are from $M$.
We may assume without loss of generality that $x\in L$, where $L\uplus R$ is the bipartition of $C$.
In step (ii) we have marked $k+3$ vertices in $L$ that have the same adjacency to $F \cap M$ as $x$; let them be $Q$.  By Prop.~\ref{lem:complete-to-triangle-free}, every vertex in $Q\cup\{x\}$ is adjacent to all vertices in $R$; on the other hand, no vertex in $Q\cup\{x\}$ is adjacent to any vertex in another component of $G - M$ different from $C$.  Therefore, all vertices in $Q\cup\{x\}$ have the same adjacency to $F\setminus L$ in $G$.  Since $|E_-| \le k$, the adjacency between $F\setminus \{x\}$ and at least one vertex in $Q$ is unchanged (noting that $|Q\cap F| \le 2$).  This vertex forms a paw with $F\backslash \{x\}$ in $G' - E_-$, a contradiction.

In the rest, $F$ contains two deleted vertices $x$ and $y$. If $x$ and $y$ are adjacent, then they are from the different parts of some component $C =  L \uplus R$. Without loss of generality, we assume that $x \in L$ and $y \in R$.
Since $|F\cap M| \le 2$, by step (ii), we can find two set $Q_1 \subseteq L$ and $Q_2\subseteq R$ that have the same adjacency to $F\cap M$  as $x$ and $y$ respectively.  Note that $|Q_1| \ge k+3$ and $|Q_2| \ge k+3$. Each vertex in $Q_1$  has the same adjacency to $F\setminus \{x\}$. The situation is similar for $Q_2$ and $F\setminus \{y\}$.  For $i =1, 2$, since $|E_-| \le k$ and $|Q_i\cap F|\le 2$, the adjacency between $F\setminus\{x\}$ and at least one vertex in $Q_i$ is unchanged.  These two vertices form a paw with $F\backslash \{x,y\}$ in $G'-E_-$, a contradiction (because $Q_1\uplus Q_2$ is complete bipartite).
Now that $x$ and $y$ are not adjacent, then they are in the same part or in different components.
Then one of $x$ and $y$ is the degree-one vertex of $F$ and the other is a degree-two vertex of $F$, and we can get that the adjacency of $x$ and $y$ to $F\cap M$ are different. By Prop.\ref{lem:complete-to-triangle-free}, the component(s) containing $x$ and $y$ is complete bipartite, then $x$ and $y$ are adjacent to all vertices in the part that does not contain them in corresponding component. By step (ii), we can find two set $Q_1$ in the part containing $x$ and $Q_2$ in the part containing $y$ that have the same adjacency to $F\cap M$ as $x$ and $y$ respectively. Then $Q_1\neq Q_2$. Since $|E_-| \le k$, $| Q_1\cap F| \le 2$ and  $|Q_2\cap F| \le 2$,  at least one vertex in $Q_1$ and at least one vertex in $Q_2$ are unchanged. These two vertices form a paw with $F\backslash \{x,y\}$ in $G'-E_-$, a contradiction.
\end{proof}

\begin{lemma}\label{lem_nummrk2}
  After Rule~\ref{rule:type-1} is applied, there are at most $O(k^4)$ vertices in all the \tone{} triangle-free components of $G - M$.
\end{lemma}
\begin{proof}
Let $C$ be any \tone{} triangle-free component. For each nonempty vertex set $S'$ with size at most two  in $M$, we mark at most $2^1 \cdot 2^2 \cdot (k+1)$ vertices in $L$ (resp., $R$) of $C$. There are at most $\binom{4k}{1} $ + $\binom{4k}{2} $ these sets in $M$.  The number of marked vertices of $C$ is bounded by $O(k^3)$. By Corollary~\ref{lem_M-partition}, there are at most $4k$ \tone{} triangle-free components. In step (i), there are at most $O(k^4)$ marked vertices. Thus, The number of marked vertices of the \tone{} triangle-free components is bounded by $O(k^4)$.
\end{proof}

Finally, we deal with \ttwo{} triangle-free components of $G- M$.

\begin{reduction}\label{rule:type-2}
Let $\cal C$ be all the \ttwo{} triangle-free components of $G - M$, and let $U = \bigcup_{C\in \cal C} V(C)$.
  \begin{enumerate}[(i)]
  \item For each $S\subseteq M$ with $|S| = 3$ and each $S'\subseteq S$, mark $k+1$ vertices from $\{x\in U\mid N(x)\cap S = S'\}$.
  \item Mark all the vertices in non-trivial components that form a triangle with $M$, and for each of them, mark $k + 1$ of its neighbors in $C$.
  \end{enumerate}
  Delete all the unmarked vertices from $C$.
\end{reduction}

\begin{lemma}\label{lem_mrk_4}
  Rule~\ref{rule:type-2} is safe.
\end{lemma}
\begin{proof}
Let $G'$ be the graph obtained after applying Rule~\ref{rule:type-2}. If $(G,k)$ is a yes-instance, then $(G',k)$ is a yes-instance. For the other direction, suppose that $(G',k)$ is a yes-instance, with a solution $E_-$.  We prove by contradiction that $G - E_-$ is paw-free as well.  A paw $F$ in $G - E_-$ contains at least one deleted vertex since $G' - E_-$ is paw-free.

By the definition of \ttwo{} triangle-free components, no triangle contains an edge in $\cal C$, implying that  the triangle $t$ in $F$ contains no edge in $\cal C$.  Note that if $F$ contains three vertices in $\cal C$, then  $t$ must contain an edge in $\cal C$, a contradiction.  If $F$ contains precisely one vertex $v$ in $\cal C$, then  by step (i), we can find a vertex $v'$ in $G'- E_-$ such that $v'$ has the same adjacency to $F\cap M$ as $v$, implying that $F\setminus \{v\} \cup \{v'\}$ in $G' - E_-$ forms a paw. If $F$ contains two vertices $x$ and $y$ in $\cal C$ , then either $x$ or $y$ is in a  triangle $t$ of $F$. Without loss of generality, we assume that $x$ is in $t$, implying that $x$ is marked in step (ii).  If $y$ is adjacent to $x$, then by step (ii), there are $k+1$ marked vertices adjacent to $x$; let them be $Q$.  The vertices in $Q$ are not adjacent to any vertex in $F\cap M$ since no triangle in $G$ contains an edge in $\cal C$. Then, each vertex in $Q$ forms a paw with $F\setminus \{y\}$ in $G'$. Since $|E_-| \le k$, there is a vertex $v'$ in $Q$ forms a paw with $F\setminus \{y\}$ in $G' - E_-$, a contradiction. If $x$ is not adjacent to $y$, by step (i), there are $k+1$ marked vertices $Q'$ having the same adjacency to $F\cap M$ as $y$ such that  each vertex in $Q'$ is not adjacent to $x$ since no triangle in $G$ contains an edge in $\cal C$. Then, each vertex in $Q'$ forms a paw with $F\setminus \{y\}$ in $G'$. Since $|E_-| \le k$, there is a vertex $v'$ in $Q'$ forms a paw with $F\setminus \{y\}$ in $G' - E_-$, a contradiction. This concludes the proof.
\end{proof}

\begin{lemma}\label{lem_numoft2}
  After Rule~\ref{rule:type-2} is applied, there are at most $O(k^4)$ vertices in all the \ttwo{} triangle-free components of $G - M$.
\end{lemma}
\begin{proof}
In step (i), there are at most $O(k^4)$ marked vertices.  For an edge $e$ in $M$,   we argue that  at most $k$ marked vertices not in trivial components, say $Q$, can form a triangle with $e$. Suppose that $|Q|>k$. Since they are in non-trivial components, for each vertex $v$ in $Q$, there is a vertex $v'$ adjacent to $v$.  By the definition of the  \ttwo{} triangle-free components, $v'$ is not incident to $e$ in $G$. Then, there is a paw formed by $e$, $v'$ and $v$. Furthermore, any two vertices in $Q$ participate two paws sharing only $e$, implying that there are $k+1$ paws sharing only $e$, contradicting that  Rule~\ref{rdt_2} is applied exhaustively. Since there are at most $\binom{4k}{2}$ edges in $G[M]$, the number of vertices marked by step (ii) is at most $\binom{4k}{2} \cdot k \cdot (k+2) $. This concludes the proof.
\end{proof}

\begin{proof}[Proof of Theorem~\ref{thm:deletion}]
  Let $(G,k)$ be a yes-instance on which none of Rules~\ref{rdt_2}--\ref{rule:type-2} is applicable, and let $M$ be the modulator.
  By Lemma~\ref{lem_numofc.m.}, the number of vertices in the complete multipartite components of  $G-M$ is bounded by $O(k^3)$. By Lemmas~\ref{lem_nummrk2} and~\ref{lem_numoft2}, the number of vertices in the triangle-free components of $G-M$ is bounded by $O(k^4)$.  Thus, the total number of vertices is $O(k^4)$.
\end{proof}

\end{document}